\documentclass[10pt, journal,twocolumn]{IEEEtran}
 \pdfoutput=1

\usepackage{amsmath}
\usepackage{mathtools}
\usepackage{amsmath}
\usepackage{amssymb}
\usepackage{cite}
\usepackage{array}
\usepackage{amsthm}
\usepackage{amsmath}
\usepackage{tabulary}
\usepackage{multirow}
\usepackage{hyperref}
\usepackage{subfig}
\usepackage{enumerate}
\IEEEoverridecommandlockouts
\usepackage{pifont}
\usepackage{epsfig}
\usepackage{graphicx}
\usepackage{url}
\usepackage{amssymb}

\graphicspath{ {./figs/} }

\newcommand{\expect}[1]{{\mathbb{E}\left[{#1}\right]}}
\newcommand{\pexpect}[1]{{\mathbb{E}^o\left[{#1}\right]}}
\newcommand{\cexpect}[2]{{\mathbb{E}_{#2}\left[{#1}\right]}}
\newcommand{\pr}{{\mathbb{P}}}
\newcommand{\ppr}{{\mathbb{P}^o}}

\newcommand{\set}[1]{{\mathcal{#1}}}

\newcommand{\ple}[1]{{\alpha_{#1}}}

\newcommand{\power}[1]{\mathrm{P}_{#1}}

\newcommand{\dnsty}[1]{{\lambda_{#1}}}
\newcommand{\tdnsty}[1]{{\tilde{\lambda}_{#1}}}

\newcommand{\R}{{\mathbb{R}}}

\newcommand{\uth}[1]{#1^\text{th}}

\newcommand{\SINR}{\mathtt{SINR}}
\newcommand{\SIR}{\mathtt{SIR}}

\newcommand{\assocr}{\mathcal{C}}
\newcommand{\cc}{\kappa}
\newcommand{\ccg}{\zeta}

\newcommand{\PPP}[1]{\Phi_{#1}}
\newcommand{\tPPP}[1]{\tilde{\Phi}_{#1}}

\newcommand{\mk}{M}

\newcommand{\indic}{1\hspace{-2mm}{1}}
\newcommand{\real}[1]{\mathbb{R}^{#1}}

\newcommand{\tmap}[1]{J(#1)}

\newcommand{\passoc}{\mathcal{A}}

\newcommand{\pkexpect}[2]{{\mathbb{E}^{o,#2}\left[{#1}\right]}}
\newcommand{\pkpr}[1]{{\mathbb{P}^{o,#1}}}
\newcommand{\om}{\omega}

\newtheorem{cor}{Corollary}
\newtheorem{lem}{Lemma}
\newtheorem{prop}{Proposition}
\theoremstyle{definition}
\newtheorem{definition}{Definition}

\theoremstyle{remark}
\newtheorem{rem}{Remark}

\begin{document}
\title{ On Association Cells in Random Heterogeneous Networks}

\author{Sarabjot Singh, Fran\c{c}ois Baccelli, and Jeffrey G. Andrews 
\thanks{This work has been supported by the Intel-Cisco Video Aware Wireless Networks (VAWN) Program and NSF grant CIF-1016649. The authors are with Dept. of Electrical and Computer Engineering at the University of Texas, Austin  (email:  sarabjot@utexas.edu, baccelli@math.utexas.edu, and jandrews@ece.utexas.edu).}
} 
\maketitle
\begin{abstract}
Characterizing user to access point (AP) association strategies in heterogeneous cellular networks (HetNets) is critical for their performance analysis, as it directly influences the load across the network.  In this letter, we introduce and analyze a  class of association strategies, which we term \textit{stationary association}, and the resulting association cells. For random HetNets, where APs are  distributed  according to a stationary point process, the area of the resulting association cells are shown to be the marks of the corresponding point process.  Addressing the need of quantifying the load experienced by a typical user, a ``Feller-paradox" like relationship is established between the area of the  association cell containing origin and that of a typical association cell.   For the specific case of Poisson point process and max power/$\SINR$ association, the mean association area of  each  tier is derived and  shown to  increase with  channel gain variance and decrease in the path loss exponents  of the corresponding tier.
\end{abstract}

\section{Introduction} 
Densification of wireless cellular infrastructure through deployment of low power APs is a promising approach to meet increasing wireless traffic demands. This complementary infrastructure  consists of various classes of APs differing in transmit powers, radio access technologies,  backhaul capacities, and deployment scenarios.  This increasing heterogeneity  and density in wireless networks  has provided an impetus  to develop new models  for their analysis and design. 
This  paper is primarily aimed to  analyze such random heterogeneous networks (HetNets), where the AP locations are modeled by a  stationary point process.

Using Poisson point processes (PPP) for modeling the irregular AP locations has been shown to be a tractable and accurate approach for characterizing signal-to-interference ratio ($\SIR$) distribution \cite{AndGanBac11,BlaKarKee12,dhiganbacand12}. 
 Although $\SIR$ is one of the key performance metrics for user performance, managing \textit{load}  or the number of users sharing the available resources per AP plays an important role in realizing the capacity gains in HetNets \cite{AndLoadCommag13}.
The load at an AP is  dictated by the user to AP  association strategy adopted in the network.  For example, users associating to their nearest AP leads to association cells conforming  to a Voronoi tessellation with AP locations as the cell centers and identical load distribution across the APs. However, in HetNets, it is desirable to  incorporate the differing base station (BS) capabilities among the classes/tiers of BSs and the propagation environments  in the association strategy. Being able to  characterize the resulting complex association cells  is one of the goals of this paper.

In this paper, we introduce \textit{stationary} association strategies, which lead to the formation of stationary association cells. Such strategies form a wide class and encompass all association patterns that are invariant by translation, including the earlier studied max $\SINR$ association \cite{MadHCN12}.  Leveraging the theory of stationary partitions introduced in \cite{Last2006},  we establish a ``Feller-paradox" like relation between the association area of the AP containing the origin to that of a typical AP in a HetNet setting, wherein the former is an \textit{area-biased} version of the latter. Such a relation has important practical implications in analyzing the load experienced by a typical user which is served, as we shall see, by an \textit{atypical} AP.  The developed theoretical framework also provides rigorous proofs for the arguments used in \cite{SinDhiAnd13,SinAnd13} for load characterization. Further, using the PPP assumption and max-power association, it is shown that the association area of a typical AP of a tier increases with the channel gain variance and decrease in the path loss exponent for the corresponding tier.

\section{Stationary association}\label{sec:StatAssoc}
The locations of the base stations are   $\{ T_n\}$  and seen as the atoms of a stationary point process (PP)  $\PPP{}$  defined on a measurable space ($\Omega,\set{A},\pr$) and  having intensity $\dnsty{}$. The analysis in this paper is for $\real{2}$ due to the practical implications, but it also extends to $\real{d}$. Further $\PPP{}$ is assumed to be $\theta_t$ compatible, where $\theta_t$ is a measurable flow on $\Omega$, so that 
\begin{equation*}
\PPP{}(\omega, B+x)=\PPP{}(\theta_x\omega, B)\,\,, \omega \in \Omega, x \in \R^2, B \in \mathcal{B},
\end{equation*}
where $\mathcal{B}$ denotes the Borel $\sigma$-field on $\R^2$.
The operation $\theta_x\om$ can also be thought of as $\PPP{}(\om)$ shifted by $-x$. 
Let $\ccg(x) \in \PPP{}\,\, \forall x \in \real{2}$ denote the base station to which a user lying at $x$ associates. The mapping $\ccg: \Omega\times\real{2}\to \real{2}$  is referred to as an \textit{association strategy}.
 \begin{definition}\label{def:cc}
\textit{Stationary Association}: An association strategy $\ccg(x)$ is stationary if the association is translation invariant, i.e., 
\begin{equation}\label{eq:ccdef}
\ccg(x) = \ccg(0)\circ\theta_x \,\, \forall x \in \real{2},
\end{equation}
where $\circ$ denotes the composition operator.
\end{definition}
Further a collection of fields $\{\mk_n(y)\} \in \real{+}\cup \infty$ $\forall\, y \in \real{2}$ is assumed associated with the atoms  $\{T_n\}$ of $\PPP{}$ such that
\begin{equation}
\begin{aligned}\label{eq:markdef}
\mk_0(y)\circ \theta_{T_n} & =  \mk_n(y+T_n) \text{ and }  \\
\mk_n(y)& = \infty \text{ if } y=T_n ,\\
\end{aligned}
\end{equation}
and therefore for a given $y$, the associated field $\mk_n(y)$ forms a sequence of marks for $\PPP{}$. Define a mapping $\cc(y) \triangleq \arg \sup \mk_n(y)$, where the $\sup$ is assumed to be well defined. Thus, by  definition \ref{def:cc}  and  (\ref{eq:markdef}), $\cc(y)$ is a stationary association. 
\begin{definition} \textit{Association Cell}:
The association cell $\assocr(T_n)$ of an AP at $T_n$ is defined as 
\begin{equation*}
\assocr(T_n) = \{y \in \real{2} : \cc(y) = T_n\}
\end{equation*}
and $|\assocr(T_n)|$ is the corresponding association area.
\end{definition}

\begin{lem}\label{lem:cellmark}
Under stationary association, the area of association cells is a sequence of marks.
\end{lem}
\begin{proof} It needs to be shown that $|\assocr(T_n)|= |\assocr(T_0)\circ\theta_{T_n}|$. 
\begin{align*}
&\assocr(T_0) \circ \theta_{T_n}=\{y : \mk_0(y) \circ \theta_{T_n} >  \mk_m(y)\circ\theta_{T_n} \,\, \forall m \neq 0\big\}\\
&\overset{(a)}{=} \{y : \mk_n(y+T_n) >   \mk_m(y+T_n) \,\, \forall m' \neq n\}\\
&=\assocr_n - T_n,
\end{align*} 
where (a) follows from  (\ref{eq:markdef}). Since area  is translation invariant the result follows.
\end{proof}
Below are listed certain   strategies that qualify as a stationary association under certain conditions.
\begin{enumerate}[I]
\item Max power association: User connects to the base station from which it  receives the maximum power.  Letting $P(n)$ denote the transmit power of AP at $T_n$, $H_n(y)$ denote the  channel power gain and  a power law path loss with  path loss exponent  $\ple{n}$, then the serving AP is 
\begin{equation}\label{eq:maxpowerassoc}
\cc(y) = \arg\sup S_n(y)=\arg\sup  P(n)H_n(y)\|T_n-y\|^{-\ple{n}}.
\end{equation}
The field $S_n(y)$ satisfies (\ref{eq:markdef}) if $\{\ple{n}\} $ and $H_n(y)$  are a sequence of marks. Further, if $\arg \sup S_n(y)$ is well defined, then max power association is stationary\footnote{If the sum $\sum_{n\geq 1} S_n(y)$ is finite almost surely (a.s.), then there exists no accumulation  at  $\sup S_n(y)$ a.s. and hence the $\arg \sup S_n(y)$ is well defined a.s.}.
\item Max $\SIR$ association: A user connects to the base station providing the highest $\SIR$. The corresponding field is 
\begin{equation*}
S_n(y)=\frac{P(n)H_n(y)\|T_n-y\|^{-\ple{n}}}{\sum_{m\neq n}P(m)H_m(y)\|T_m-y\|^{-\ple{m}}}.
\end{equation*}
It can be  seen  that max $\SIR$ association is equivalent to max power association in I. Note that the association cells formed in this case are different than the $\SINR$ coverage cells defined in \cite{BB2001}.

\item Nearest base station association: This results in the classical case of Voronoi cells as association cells, which are stationary.
\end{enumerate}

In this paper, the probability and expectation under the Palm probability  are denoted  by $\ppr$ and $\pexpect{}$ respectively.
\begin{prop}\label{prop:inv}
For all measurable functions  $f : \Omega\to \real{+}$
\begin{equation*}
\expect{f} =  \dnsty{}\pexpect{\int_{\assocr(T_0)}f\circ\theta_{u}\mathrm{d}u}.
\end{equation*}
\end{prop}
\begin{proof} 
Stationary association $\cc$ satisfying  (\ref{eq:ccdef}) can be seen as the stationary partition introduced in  \cite{Last2006}.
The  proof follows using Theorem 4.1 of \cite{Last2006}, which applies  Mecke's formula to the function $h(\om,x)= f(\om)\indic(x\in \PPP{}(\om), \cc(0)=x)$.
\end{proof}
\begin{figure*}
  \centering
\subfloat[No channel variance,  $\sigma_1=\sigma_2 = 0$]{\label{fig:v0}\includegraphics[width=0.7\columnwidth]{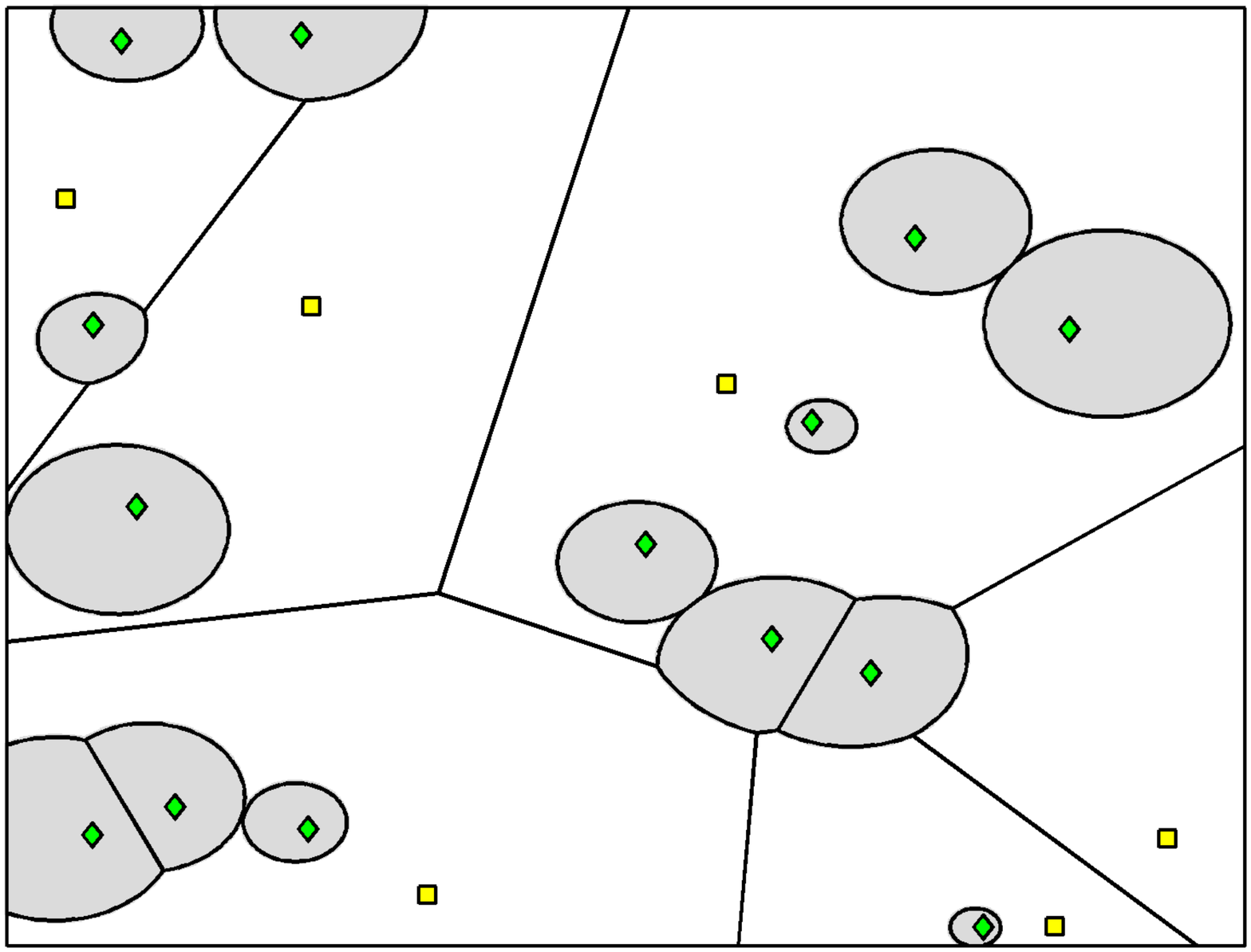}}
\subfloat[Low channel variance, $\sigma_1=1$, $\sigma_2 =1$]{\label{fig:v1}\includegraphics[width=0.7\columnwidth]{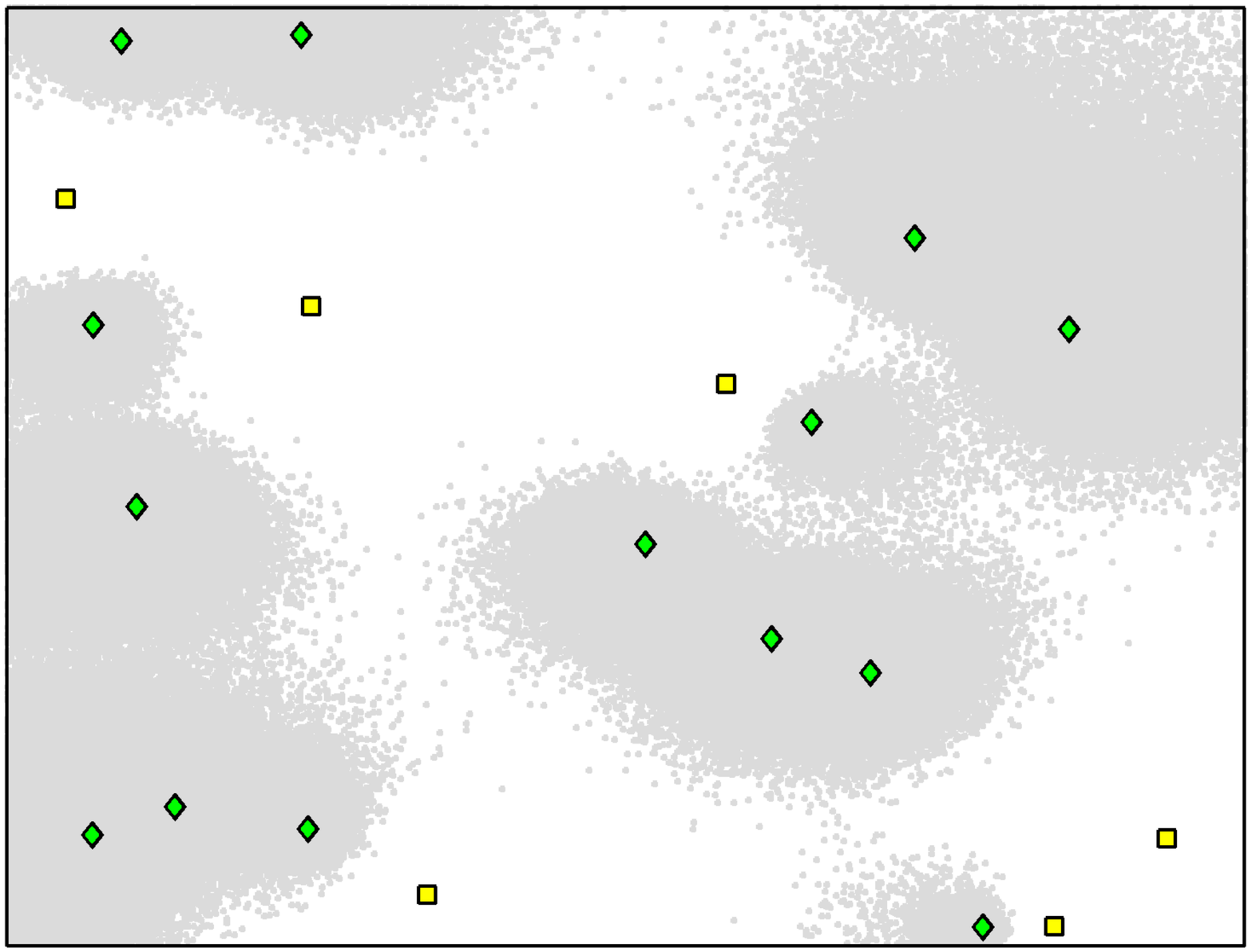}}
\subfloat[High channel variance, $\sigma_1 =1$, $\sigma_2 =2$]{\label{fig:v5}\includegraphics[width=0.7\columnwidth]{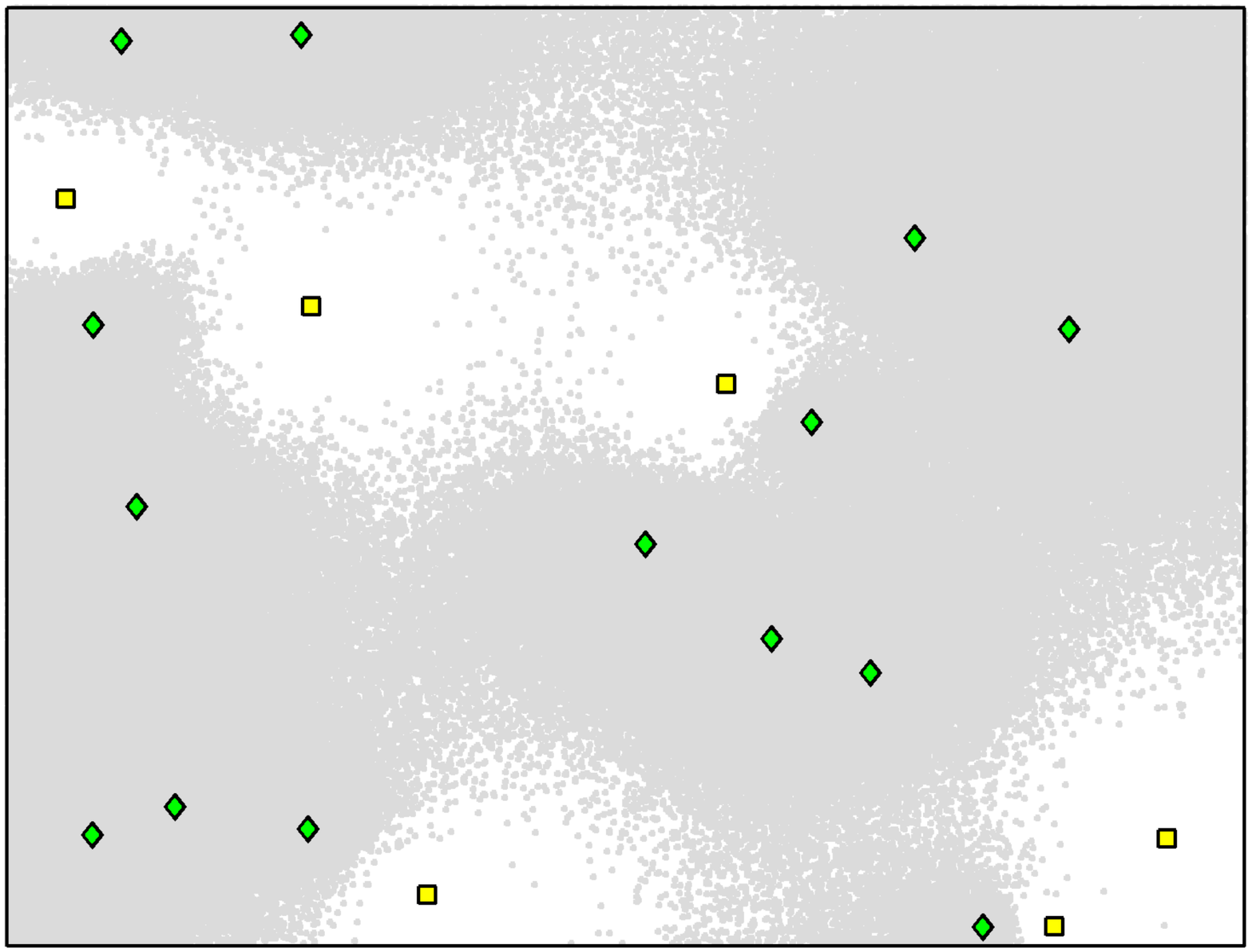}}
\caption{The shaded region is served by the APs of tier-2  (diamonds), while the rest of the area is served by tier-1 APs (squares). }
 \label{fig:assoc}
\end{figure*}

\section{Association in $K$-tier Networks}
In a $K$-tier HetNet, the APs are assumed to belong to   $K$ distinct \textit{classes}. Assuming independent deployment  of APs of different tiers,  we define i.i.d.  marks  mapping the AP index to tier index as $\tmap{T_n}\in \{1\ldots K\}$. The mapping distribution for a typical BS is 
\begin{equation*}
p_k \triangleq \ppr(\tmap{T_0}=k). 
\end{equation*}
 The location of the APs of $\uth{k}$ tier  is denoted  by the  point process $\PPP{k}$, where 
\begin{equation*}
\PPP{k} = \sum_{T_i \in \PPP{}} \delta_{T_i}\indic(\tmap{T_i}=k).
\end{equation*}
The following proposition builds up on Prop. \ref{prop:inv} to relate the probability of origin being contained in the association cell of tier $i$ to the association area of a typical cell of the corresponding tier.
\begin{prop} \label{prop:passoc}
 The probability that the origin is contained in the association cell  of an atom of $\PPP{i}$ is
\begin{equation*}
\passoc_i  \triangleq \pr(\tmap{\cc(0)}= i) = \dnsty{}p_i\pkexpect{|\assocr(T_0)|}{i}
\end{equation*}
\end{prop}
\begin{proof} 
Using Prop. \ref{prop:inv} with $f = \indic(\tmap{\cc(0)}= i)$, we obtain
\begin{align*}
\expect{\indic(\tmap{\cc(0)}= i)}  &=  \dnsty{} \pexpect{\int_{\assocr(T_0)} \indic(\tmap{\cc(0)}\circ \theta_u= i)\mathrm{d}u}\\
\pr(\tmap{\cc(0)}= i) & \overset{(a)}{=} \dnsty{} \pexpect{\indic(\tmap{\cc(0)}= i)|\assocr(T_0)|}\\
\passoc_i & \overset{(b)}{=} \dnsty{}\pkexpect{|\assocr(T_0)|}{i}\ppr(\tmap{\cc(0)}=i),
\end{align*}
where (a) holds, as under palm $\tmap{\cc(0)}\circ \theta_u = \tmap{\cc(u)}= \tmap{\cc(0)}$  for $ u \in \assocr(T_0)$ and (b) follows from Bayes theorem. 
\end{proof}
For the case where users in the network form a homogeneous PPP, $\passoc_i$ denotes the probability of a typical user associating with the $\uth{i}$ tier.
The following proposition gives a conditional form of Prop. \ref{prop:inv} in a  $K$-tier setting.
\begin{prop}\label{prop:invktier} For all measurable functions $g: \Omega \to \real{+}$
\begin{equation*}
\expect{g|\tmap{\cc(0)}=i}=  \frac{\pkexpect{\int\limits_{\assocr(T_0)}g\circ\theta_{u}\mathrm{d}u}{i} }{\pkexpect{|\assocr_0|}{i}}.
\end{equation*}
\end{prop}
\begin{proof} Using $f= g\indic(\tmap{\cc(0)}=i)$ in Prop.  \ref{prop:inv},  the LHS is
\begin{equation*}
\expect{g\indic(\tmap{\cc(0)}=i)}= \expect{g|\tmap{\cc(0)}=i}\passoc_i,
\end{equation*}
and the RHS is 
\small
\begin{equation*}
 \dnsty{} \pexpect{\indic(\tmap{\cc(0)}=i)\int_{\assocr(T_0)}g\circ \theta_u  \mathrm{d}u} =  \dnsty{}p_i \pkexpect{\int_{\assocr(T_0)}g\circ \theta_u  \mathrm{d}u}{i}
\end{equation*}\normalsize
Using   Prop. \ref{prop:passoc} in the above, gives the result.
\end{proof}
Using the above proposition, the density of association area  of the AP of tier $i$  containing  origin (assuming it exists)   can be given in terms of that  of the area of a typical association cell of  the corresponding tier.
\begin{cor}\label{cor:lawtypical}
The density of the area of the association cell of tier $i$ containing origin is given by
\begin{equation*}
f_{|\assocr(\cc(0))|}^i(c) = \frac{c f^{o,i}_{|\assocr(T_0)|}(c)}{\pkexpect{|\assocr(T_0)|}{i}},
\end{equation*}
where $ f^{o,i}_{|\assocr(T_0)|} $ is the density of the  area of  a typical association cell of tier $i$.
\end{cor}
\begin{proof} Using $g = \indic(v\leq|\assocr(\cc(0))|\leq v+\mathrm{d}v)$  in Prop. \ref{prop:invktier}  we get
\begin{multline}
\pr(c\leq|\assocr(\cc(0))|\leq c+\mathrm{d}c|\tmap{\cc(0)}=i) \\= \frac{\pkexpect{\int_{\assocr(T_0)}\indic(c\leq|\assocr(\cc(0)\circ\theta_u)|\leq c+\mathrm{d}c)\mathrm{d}u}{i}}{\pkexpect{|\assocr(\cc(0))|}{i}}\nonumber
\end{multline}
\begin{align*}
f^{o,i}_{|\assocr(\cc(0))|}(c)\overset{(a)}{=} \frac{\pkexpect{\indic(c\leq|\assocr(\cc(0))|\leq c+\mathrm{d}c)|\assocr(\cc(0))|}{i}}{\pkexpect{|\assocr(\cc(0))|}{i}},
\end{align*}
where (a) follows from the fact that under the Palm distribution $|\assocr(\cc(0)\circ\theta_u)|= | \assocr(\cc(0))|$ for $u \in \assocr(\cc(0))$. The final result is obtained  using Bayes theorem and the fact that under the Palm distribution  $\cc(0) = T_0$.
\end{proof}
As a consequence of the above corollary it can be stated that the area of the association cell containing a typical user is larger than that of  a typical cell, and the following holds
\begin{equation*}\label{eq:meanrelation}
\expect{|\assocr(\cc(0))|^d|\tmap{\cc(0))=i}}  =  \frac{\pkexpect{|\assocr(\cc(0))|^{d+1}}{i}}{\pkexpect{|\assocr(\cc(0))|}{i}}\,\, \forall d\in \real{}.
\end{equation*}

\section{Mean Association Area in PPP HetNets}
In this section, the mean association  area is derived for the case  where the base station process $\PPP{}$ is assumed to be a PPP. The general max power/$\SINR$ association given in (\ref{eq:maxpowerassoc}) is considered. It is further assumed that the APs of $\uth{k}$ tier have the same constant power and  path loss exponents, and have independent but identical channel gain distribution, i.e, $P(n) =\power{k}$, $\ple{n}=a_k$, and $H_n(y) \overset{(d)}{=} H_k$ $\forall$ $T_n \in \PPP{k}$. Due to the i.i.d. assumption on  $J$ marks,  by  the thinning theorem \cite{BacBook09}, each tier process $\PPP{k}$ is a PPP with density $\dnsty{k} \triangleq p_k\dnsty{}$ for $k=1\ldots K$. For illustration, Fig. \ref{fig:assoc} shows the association cells in a two tier setup with $\power{1}=53$ dBm, $\power{2}= 33$ dBm, $a_1=a_2=4$, and the channel gain is lognormal $H_k \sim \ln \mathcal{N}(0,\sigma_k)$. As seen from the plots, increasing the variance in the channel gain for the second tier increases the corresponding association areas (the shaded areas). This observation  is made rigorous by the following analysis.
\subsection{Analysis}
\begin{lem}
Under the max power association, the mean association area of  a typical base station of the $\uth{i}$ tier  is 
\small
\begin{equation*}
2\pi \int_{0}^{\infty} r\cexpect{\exp\left(-\pi\sum_{k=1}^K \tdnsty{k} r^{2a_i/a_k}\power{i}^{-2/a_k} H_i^{-2/a_k}\right)}{H_i} \mathrm{d}r,
\end{equation*}\normalsize
where $\tdnsty{k} = \dnsty{k}\power{k}^{2/a_k}\expect{H_k^{2/a_k}}$ and $\expect{H_k^{2/a_k}} < \infty$.
\end{lem}
\begin{proof} The mean association area of a typical cell of the $\uth{i}$ tier is $
\pkexpect{|\assocr(T_0)|}{i}  = \int_{\real{2}}\pkpr{i}\left(u\in\assocr(T_0)\right)\mathrm{d} u $, which $ \forall n\neq 0$
\small
\begin{align*}
&=  \pkpr{i}\Bigg(\|u\|^{a_i/\ple{n}} (H(0,u)\power{i})^{-1/\ple{n}} <  \|T_n -u\| (H(n,u)\power{n})^{-1/\ple{n}}\Bigg)\\
& = \pkpr{i}\left(\bigcap_{k=1}^K\tPPP{k}\left(B^o(0, \|u\|^{a_i/a_k}(H(0,u)\power{i})^{-1/a_{k}}\right)=0\right),
\end{align*} 
\normalsize
where $\tPPP{k}$ denotes  the  PPP formed by  transforming the atoms of $\PPP{k}$: $T_n  \to (T_n-u)(H(n,u)\power{k})^{-1/a_k}$. By the  i.i.d. displacement theorem \cite{BacBook09}, $\tPPP{k}$ is a homogeneous PPP with $\tdnsty{k} =  \dnsty{k}\power{k}^{2/a_k}\expect{H_k^{2/a_k}}$, given $\expect{H_k^{2/a_k}} < \infty$, \cite{BlaKee13,MadHCN12}. Thus
\begin{align}
&\pkpr{i}\left(u\in\assocr(T_0)\right)= \int_{0}^\infty\pkpr{i,h}\left(u\in\assocr(T_0)\right)f_{H_i}(h)\mathrm{d}h\nonumber\\
\overset{(a)}{=}&\cexpect{\prod_{k=1}^K \pkpr{i}\left(\tPPP{k}\left(B^o\left(0,\|u\|^{a_i/a_k}(H_i\power{i})^{-1/a_k}\right)\right)=0\right)}{H_i}\nonumber\\
  =&\cexpect{\exp\left(-\pi\sum_{k=1}^K \tdnsty{k} \|u\|^{2a_i/a_k}(H_i\power{i})^{-2/a_k}\right)}{H_i}\label{eq:princell}
\end{align}
\end{proof}
For the case with the path loss exponents of each tier being the  same: $a_k \equiv a$, the mean association area simplifies to 
\begin{equation}\label{eq:marea}
\pkexpect{|\assocr(T_0)|}{i} = \frac{\power{i}^{2/a}\expect{H_i^{2/a}}}{\sum_{k=1}^K \dnsty{k}\power{k}^{2/a}\expect{H_k^{2/a}}}
\end{equation}
and thus depends on only the $\uth{\frac{2}{a}}$ moment of the channel gain. Using Prop. \ref{prop:passoc} for association probability leads to the earlier derived result in \cite{MadHCN12}, which used \textit{propagation invariance}. 

\begin{rem} The framework developed above can be used to compute additive functionals over association cells.
 Using Campbell's theorem \cite{BacBook09}, the mean  of an additive characteristic $g$ associated with a typical association cell of tier $i$ and defined on an independent PPP $\PPP{u}$ of intensity $\dnsty{u}$ is 
\begin{multline}
\bar{S}=\pkexpect{\cexpect{\sum_{j} g(Y_j)\indic(Y_j \in \assocr(T_0))}{\PPP{u}}}{i} \\
=  \int_{\real{2}} g(y)\pkpr{i}(y \in \assocr(T_0))\dnsty{u}\mathrm{d}y\nonumber.
\end{multline}
For example,   if $\PPP{u}$ represents the user point process and $g(x)=\|x\|^{-a}$ (a  path loss function), then $\bar{S}$ represents the mean total power received at  a typical AP of tier $i$ from all the users served by it and is given by (using (\ref{eq:princell}))

\small
\begin{equation*}
2\pi\dnsty{u}\int_{r > 0} r^{-a+1}\cexpect{\exp\left(-\pi\sum_{k=1}^K \tdnsty{k} r^{2a_i/a_k}(H_i\power{i})^{-2/a_k}\right)}{H_i}\mathrm{d}r
\end{equation*}
\normalsize
\end{rem}

\subsection{Numerical Results}
We consider a two tier (macro and pico, say) setup along with max power association with respective transmit powers: $\power{1} =53$ dBm and $\power{2} = 33$ dBm. The variation in the mean association area with the  variation in  density of small cells (second tier) is shown in Fig. \ref{fig:areasigmatrend}, where $\sigma_1 = 2$, $a_1 =a_2=3.5$,  $\dnsty{1} = 1$ BS/sq. Km, and the channel gains are assumed lognormal with $H_k \sim \mathrm{ln} \mathcal{N}(0,\sigma_k)$. It can be seen that  with increasing variance in the channel propagation, the corresponding mean association area increases. This follows from (\ref{eq:marea}) and the fact that $\expect{H_k^{2/a}}= \exp(0.5(2/a)^2\sigma_k^2)$.
The effect of  path loss exponent on the mean association area is shown   in Fig. \ref{fig:areapletrend}. For the plot, $\sigma_1=2$, $\sigma_2 =4$, and $a_1 =3$.  As can be seen, with decreasing path loss exponent of small cells, the corresponding association area increases. Intuitively, the lower the path loss exponent, the lower the decay rate of the corresponding AP's transmission power and hence there will be a larger number of users associating with the same.
\begin{figure}
	\centering
		\includegraphics[width=0.8\columnwidth]{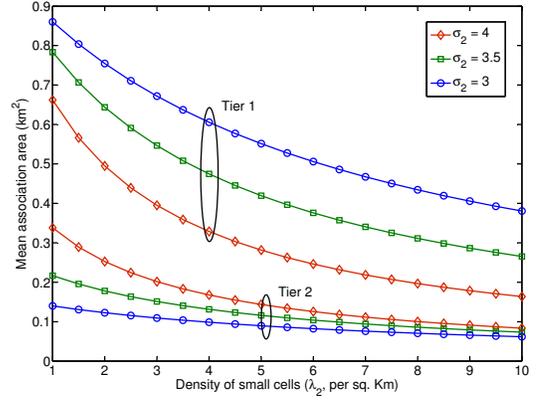}
		\caption{Variation of mean association areas of two tiers with density for different channel gain  variances of second tier}
	\label{fig:areasigmatrend}
\end{figure}
\begin{figure}
	\centering
		\includegraphics[width=0.8\columnwidth]{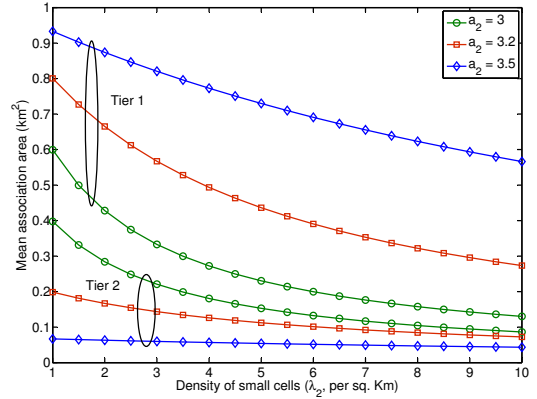}
		\caption{Variation of mean association areas of two tiers with density for different path loss exponents of second tier}
	\label{fig:areapletrend}
\end{figure}
\section{Conclusion}
We introduce the notion of stationary association for random HetNets with the resulting association areas shown to be the marks of the corresponding point process. Analogous to a Voronoi tessellation, an inversion formula relating the association area of the cell containing origin to a typical association area for each tier  is proved. It is  shown that with max power association,  the mean association area of small cells decreases with  path loss exponent and increases with channel gain variance.

\bibliographystyle{ieeetr}
\bibliography{IEEEabrv,C:/Users/Sarabjot/Documents/Dropbox/research/refoffload}

\begin{thebibliography}{10}

\bibitem{AndGanBac11}
J.~G. Andrews, F.~Baccelli, and R.~K. Ganti, ``A tractable approach to coverage
  and rate in cellular networks,'' {\em {IEEE} Trans. Commun.}, vol.~59,
  pp.~3122--3134, Nov. 2011.

\bibitem{BlaKarKee12}
B.~Blaszczyszyn, M.~K. Karray, and H.-P. Keeler, ``Using {Poisson} processes to
  model lattice cellular networks,'' in {\em Proc. {IEEE INFOCOM}},
  pp.~773--781, Apr. 2013.

\bibitem{dhiganbacand12}
H.~S. Dhillon, R.~K. Ganti, F.~Baccelli, and J.~G. Andrews, ``Modeling and
  analysis of {$K$}-tier downlink heterogeneous cellular networks,'' {\em
  {IEEE} J. Sel. Areas Commun.}, vol.~30, pp.~550--560, Apr. 2012.

\bibitem{AndLoadCommag13}
J.~G. Andrews {\em et~al.}, ``An overview of load balancing in {HetNets}: Old
  myths and open problems,'' {\em {IEEE} Wireless Commun. Mag.}, submitted.
\newblock Available at: http://arxiv.org/abs/1307.7779.

\bibitem{MadHCN12}
P.~Madhusudhanan, J.~Restrepo, Y.~Liu, and T.~Brown, ``Downlink coverage
  analysis in a heterogeneous cellular network,'' in {\em IEEE GLOBECOM},
  pp.~4170--4175, Dec. 2012.

\bibitem{Last2006}
G.~Last, ``Stationary partitions and {P}alm probabilities,'' {\em {Advances in
  Applied Probability}}, vol.~38, pp.~602--620, Sept. 2006.

\bibitem{SinDhiAnd13}
S.~Singh, H.~S. Dhillon, and J.~G. Andrews, ``Offloading in heterogeneous
  networks: Modeling, analysis, and design insights,'' {\em {IEEE} Trans.
  Wireless Commun.}, vol.~12, pp.~2484--2497, May 2013.

\bibitem{SinAnd13}
S.~Singh and J.~G. Andrews, ``Joint resource partitioning and offloading in
  heterogeneous cellular networks,'' {\em {IEEE} Trans. Wireless Commun.},
  submitted.
\newblock Available at: http://arxiv.org/abs/1303.7039.

\bibitem{BB2001}
F.~Baccelli and B.~B{\l}aszczyszyn, ``On a coverage process ranging from the
  boolean model to the {Poisson-Voronoi } tessellation with applications to
  wireless communications,'' {\em Advances in Applied Probability}, vol.~33,
  pp.~pp. 293--323, June 2001.

\bibitem{BacBook09}
F.~Baccelli and B.~Blaszczyszyn, {\em Stochastic Geometry and Wireless
  Networks, Volume I -- Theory}.
\newblock NOW: Foundations and Trends in Networking, 2009.

\bibitem{BlaKee13}
B.~B{\l}aszczyszyn and H.~P. Keeler, ``Equivalence and comparison of
  heterogeneous cellular networks,'' in {\em International WDN Workshop on
  Cooperative and Heterogeneous Cellular Networks}, Sept. 2013.

\end{thebibliography}

\end{document}